\documentclass[11pt,oneside,reqno]{amsart}
\usepackage[margin=1.2in]{geometry}

\usepackage[utf8x]{inputenc}
\usepackage[T1]{fontenc}

\usepackage{amsmath, amsthm, amssymb, esint, mathtools, mathabx, tikz-cd}
\usepackage{graphicx}
\usepackage[colorinlistoftodos]{todonotes}
\usepackage[colorlinks=true, allcolors=blue]{hyperref}
\usepackage{caption}
\usepackage{cleveref}

\DeclareMathOperator*{\range}{\textnormal{range}}
\DeclareMathOperator*{\spec}{\textnormal{spec}}

\DeclareMathOperator*{\Z}{\mathbb{Z}}
\DeclareMathOperator*{\diag}{\textnormal{diag}}
 \DeclareMathOperator{\spn}{span}

\theoremstyle{thmstyletwo}%
\newtheorem{theorem}{Theorem}[section]%  meant for continuous numbers
\newtheorem{proposition}[theorem]{Proposition}%
\newtheorem{corollary}[theorem]{Corollary}

\newtheorem{example}{Example}[section]%
\newtheorem{remark}{Remark}
\newtheorem{lemma}[theorem]{Lemma}%
\newtheorem{question}[theorem]{Question}%
\newtheorem{definition}{Definition}

\def\hdmd{\odot}

\numberwithin{equation}{section}

\begin{document}

\title{Convolutional dynamical sampling and some new results}

\author{Longxiu Huang} 
\address{
Department of Computational Mathematics Science and Engineering and Department of Mathematics, Michigan State University, USA}
\email{huangl3@msu.edu}

\author{A. Martina Neuman}
\address{Department of Computational Mathematics Science and Engineering, Michigan State University, USA}
\email{neumana6@msu.edu}

\author{Sui Tang}
\address{Department of Mathematics, University of California Santa Barbara, USA}
\email{suitang@math.ucsb.edu}

\author{Yuying Xie}
\address{Department of Computational Mathematics Science and Engineering and Department of Statistics and Probability, Michigan State University, USA}
\email{xyy@msu.edu}

 \begin{abstract} In this work, we explore the dynamical sampling problem on $\ell^2(\mathbb{Z})$ driven by a convolution operator defined by a convolution kernel. This problem is inspired by the need to recover a bandlimited heat diffusion field from space-time samples and its discrete analogue. In this book chapter, we review recent results in the finite-dimensional case and extend these findings to the infinite-dimensional case, focusing on the study of the density of space-time sampling sets.
\end{abstract}

\maketitle

\vspace{-25pt}

\section{Introduction} \label{sec:intro}

Dynamical sampling, inspired by the works of \cite{lu2009spatial, lu2011localization}, operates on the principle that spatial sampling limitations can be offset by utilizing temporal evolution during data recovery \cite{aldroubi2013dynamical}. The time-dependent characteristics of signal evolution, driven by an external force, enhance the quality of the collected samples \cite{huang2021robust}, distinguishing dynamical sampling from traditional static methods \cite{candes2006stable, sun2007nonuniform, nashed2010sampling}. This suggests that a space-time trade-off is practical in many industrial settings where the limited availability of sampling devices can be compensated by their increased use, especially where cost or space constraints exist. The relevant works include the mobile sampling method developed in \cite{unnikrishnan2012sampling, grochenig2015minimal, ulanovskii2021reconstruction, zlotnikov2022planar, jaye2022sufficient}, which applies the space-time trade-off principle in time-varying bandlimited fields. In these methods, mobile sensors move along continuous paths through space, taking measurements along their trajectories. And wireless sensor networks represent a key application of this concept \cite{hormati2009distributed,lu2009spatial,  reise2010reconstruction, ranieri2011sampling,  reise2012distributed}.

The mathematical framework of dynamical sampling has been thoroughly developed in recent literature and is divided into two main categories. The first category addresses the characterization of feasible space-time samples that can enable recovery by linking the spectral properties of driving operators with sampling locations. Initial results in this area were developed for convolution-driven operators in discrete time scenarios \cite{aldroubi2013dynamical,aldroubi2015exact, aceska2015multidimensional}, with a complete characterization of the universal sampling set presented in \cite{tang2017system}. A general framework for bounded linear operators was subsequently established \cite{aldroubi2017dynamical, cabrelli2020dynamical}, which included extensions to phaseless sampling \cite{aldroubi2020phaseless,beinert2023phase}, random sampling with low-rank constraints \cite{huang2021robust,yao2023space}, and continuous-time sampling \cite{aldroubi2019frames,aldroubi2021sampling}. A comprehensive survey on this topic is available in \cite{christensen2021survey}. The second category focuses on recovering driving operators from space-time samples. This includes methods for convolution operators and their generalization \cite{aldroubi2016krylov,tang2017system,aldroubi2018dynamical} to the general matrix case using Prony-type methods and their stabilized variants \cite{cheng2021estimate}. More recent studies \cite{kummerle2022learning} have employed a lifting technique by transforming the problem into a low-rank block Hankel matrix completion challenge.

A central challenge in dynamical sampling is how to choose space-time sampling locations so that the recovery process is stable. Although a general tool via spectral theory has been developed, more refined results are greatly anticipated, particularly in connecting the density theorem for classical sampling problems with dynamical sampling. Given that convolution underpins the physics of diffusion processes, many researchers have focused on temporal evolution driven by a convolution operator. A recurring theme is that simple sampling across a discrete subgroup may be inadequate for signal reconstruction. Instead, a careful balance between sampling density and kernel profile is necessary for effective sampling. For instance, \cite{aldroubi2015exact} introduced a convolution kernel construction that renders spatiotemporal sampling on a proper subgroup \( m\mathbb{Z} \) of \( \mathbb{Z} \) unstable when \( m \geq 2 \), necessitating additional initial sampling. Similarly, \cite{tang2017universal} determined the necessary extra sample size for effective sampling on any subgroup of a finite Abelian group \( \mathbb{Z}_{N} \), and \cite{aceska2015multidimensional} applied a similar approach to multidimensional torsion groups. These studies emphasize the critical interplay between the convolution kernel and sampling density, a dynamic also highlighted in \cite{aldroubi2021sampling} in the continuous domain, where the heat kernel profile dictates density thresholds on \( \mathbb{R} \). Motivated by these insights, our study examines the intricate relationship between stability, spatial sampling density, and kernel profile during an evolution process on \( \mathbb{Z} \), and explores the inherent nature of this stability. To begin, we formulate our problem as follows.

\subsection{Problem statement}
In dynamical sampling, data are collected repeatedly over time from the fixed spatial locations. This sampled data  is represented by:
\begin{equation}
\label{eqn:general_DS}
S(f) = \{\langle A^{s}f, g\rangle : g \in \mathcal{G}, s \in T\},
\end{equation}
where $f \in \mathcal{H}$ is the initial signal, $A$ is the driving operator, $\mathcal{G} \subset \mathcal{H}$ denotes the set of spatial locations, and $T$ represents the sampling times. The goal of the dynamical sampling problem is to derive necessary and/or sufficient conditions that involve the driving operator $A$, the spatial  set $\mathcal{G}$, and the sampling time set $T$, ensuring that the sampling set $S(f)$ sufficiently (and stably) reconstructs the initial signal $f$.

In this study, we focus on the scenario where $\mathcal{H}=\ell^2(\mathbb{Z})$, $\mathcal{G} \subseteq \{e_i: i\in \mathbb{Z}\}$ the standard basis for $\ell^2(\mathbb{Z})$, $T \subseteq\mathbb{N}_0 := \{0, 1, 2, \cdots\}$, and $A$ is a convolution operator on $\ell^2(\mathbb{Z})$ characterized by the kernel $a = \{a(k)\}_{k \in \mathbb{Z}} \in \ell^1(\mathbb{Z})$. This setting leads us to formulate the following question:
\begin{question}
\label{question:orignal}
Under what conditions on the spatial  set $\mathcal{G}$ and the driving operator $A$, represented by the kernel $a$, should the sampling set
\begin{equation}
\label{sampset}
S(f) = \{\langle A^{s}f, g\rangle : g \in \mathcal{G}, s \in T\}
\end{equation}
be sufficient for one to (stably) recover the initial state $f$?
\end{question}

Due to the relationship between $A$ and its adjoint $A^*$, Question~\ref{question:orignal} is equivalent to the question of under what condition(s) on $\mathcal{G}$ the set $\left\{(A^{*})^sg:g\in\mathcal{G}, s\in \mathbb T\right\}$ is complete or forms a frame for $\ell^2(\mathbb{Z})$. To simplify notation, we focus on answering the following question throughout this work.
\begin{question} \label{question}
Under what condition(s) on $\mathcal{G}\subseteq\{e_i: i\in \mathbb{Z}\}$ and  the driving convolution operator $A$,  does the set 
\begin{equation}\label{eqn:sensing_set}
   \Omega=\left\{A^{s}g:g\in\mathcal{G}, s\in T \right\}
\end{equation}
form a complete set or a frame for $\ell^2(\mathbb{Z})$?
\end{question}

\subsection{Review on Finite Dimensional Case}

The finite dimensional version of  Question~\ref{question} was studied in \cite{aldroubi2013dynamical, aceska2015multidimensional}, where the sampling set is a sub-lattice of \(\mathbb{Z}^d\) and their cosets. The major tool used in these studies is the Poisson summation formula, which transforms the sampling problem into a linear system. Analysis of the null space of the corresponding system allows for the construction of stable sampling sets. This approach was fully developed in \cite{tang2017universal}, where universal sampling sets for a class of convolution operators were determined based on their eigenvalues and symmetric properties.

The dynamical sampling problem in the infinite dimensional case is more challenging \cite{aldroubi2015exact, aldroubi2017dynamical}.  Although some results exist when the driving operators are self-adjoint (or more generally,  normal operators), convolution operators deserve  special attention  since more refined results are possible for them. 
 
\subsection{Contributions and organizations} 

The objective of this study is to analyze the interplay between the sampling set $\mathcal{G}$ and the driving operator $A$ to determine whether the set $\Omega$ defined in \eqref{eqn:sensing_set} can form a complete or frame set for $\ell^2(\mathbb{Z})$. Our first result, summarized in Theorem 1, demonstrates that $\mathcal{G}$ must be infinite to ensure $\Omega$ forms a frame for $\ell^2(\mathbb{Z})$. This finding parallels the case where $A$ is a bounded self-adjoint operator, as studied in \cite{aldroubi2019frames}. Further details are provided in \Cref{sec:finitecard}.

We then examine a sampling set $\mathcal{G}$ of a sub-lattice form, specifically $\mathcal{G} = \{e_k : k \in \Lambda\}$ with $\Lambda = \{mj + c : j \in \mathbb{Z}, c = 0, 1, \ldots, L-1\}$ for some $m \geq 2$ and $1 \leq L \leq m$, referring to non-uniform periodic sampling in the literature. In this setup, we provide a characterization of the relationship among $L$, $m$, and the kernel $a$ necessary for $\Omega$ to be complete in $\ell^2(\mathbb{Z})$. Our result builds upon findings in the finite-dimensional case (see \cite{aceska2015multidimensional, tang2017universal}). Specifically, we demonstrate the equivalence of completeness and the frame property of $\Omega$ in $\ell^2(\mathbb{Z})$ when the kernel $a \in \ell^1(\mathbb{Z})$ and $L=1$, which is non-trivial in the infinite-dimensional case. More details can be found in \Cref{sec:subsampling}.

Finally, we present a result characterizing the sampling density of $\Lambda$ required for stable sampling and reveal how it depends on the convolution kernel $a$. We show that the maximal gap between spatial samples cannot be arbitrarily large. This investigation draws inspiration from \cite{aldroubi2021sampling}, which studied continuous-time heat diffusion in a bandlimited space, as described in \eqref{eqn:general_DS}. For more detailed results, refer to \Cref{sec:density}.

\section{Preliminary and Notation} \label{sec:notation}
Throughout this paper, the symbol $[N]$ denotes the set of integers $\{0, 1, \dots, N-1\}$. We employ the notation $|\cdot|$ in multiple ways: it can indicate the absolute value of a scalar,  the cardinality of a set, or the length of an interval, depending on the context. This work focuses on  the Hilbert space $\ell^2(\mathbb{Z})$, equipped with the inner product defined as:
$\langle f,g\rangle  := \sum_{k\in\mathbb{Z}} f(k)\overline{g(k)}$. The convolution operator $A: \ell^2(\mathbb{Z}) \to \ell^2(\mathbb{Z})$, which employs a kernel $a = \{a(k)\}_{k \in \mathbb{Z}} \in \ell^1(\mathbb{Z})$, is defined by:
\begin{equation} \label{def:convop}
A(f) (j) := (a \ast f) (j) := \sum_{k\in\mathbb{Z}}a(k)f(j-k), \textnormal{for}  f \in \ell^2(\mathbb{Z})~\textnormal{and}~j\in\mathbb{Z}.
\end{equation}
 Successive applications of $A$ produce a series evolution of $f$, expressed as:
\begin{equation} \label{as}
   A^{s}f = (\underbrace{a\ast\cdots\ast a}_{\textnormal{$s$ times}})\ast f =: a^{(s)}\ast f,
\end{equation}
where $s \in \mathbb{N}_0 := \{0, 1, 2, \cdots\}$. 
To assist the reader, we will introduce some fundamental concepts below that are pivotal to our analysis. The Fourier transform  plays a central role in our analysis, particularly the one on $\ell^2(\mathbb{Z})$. 

\begin{definition}[Fourier transformation]
   %  For a function $f \in L^2(\mathbb{R})$, the Fourier transform $\mathcal{F}$ evaluated at $\xi$ is defined as follows:
% \begin{equation} \label{FTR}
% \mathcal{F}(f)(\xi) := \int_{\mathbb{R}} f(x)e^{-i2\pi x\xi} dx \in L^2(\mathbb{R}).
% \end{equation}
%   The corresponding inverse Fourier transform $\mathcal{F}^{-1}$ for $g\in L^2(\mathbb{R})$   is defined as:
% \begin{equation*}
%    \mathcal{F}^{-1}(g)(x)=\int_{\mathbb{R}}g(\xi)e^{i2\pi x\xi}\,d\xi \in L^2(\mathbb{R}).
% \end{equation*}
% For example, 
% \begin{equation} \label{sinc}
%      \sinc(x):=\frac{\sin(\pi x)}{\pi x} =
%      \int_{-1/2}^{1/2}e^{i2\pi\xi x}\,d\xi=\mathcal{F}^{-1}(\chi_{[-1/2,1/2)})(x).
% \end{equation}
The Fourier transform on $\ell^2(\mathbb{Z})$, denoted by $\mathcal{F}_{\mathbb{Z}}$, is defined as:
\begin{equation} \label{FTZ}
\mathcal{F}_{\mathbb{Z}}({f})(\omega) = \hat{f}(\omega) := \sum_{k \in \mathbb{Z}} f(k) e^{-i2\pi k \omega} \in L^2(\mathbb{T}).
\end{equation}
For instance, if $a \in \ell^1(\mathbb{Z})$, then $\hat{a}$ belongs to $C(\mathbb{T})$ and is also a subset of $L^2(\mathbb{T})$. Its inverse Fourier transform $\mathcal{F}_{\mathbb{Z}}^{-1}$, operating on $L^2(\mathbb{T})$, is given by:
\begin{equation*}
\mathcal{F}_{\mathbb{Z}}^{-1}(g)(k) = \check{g}(k) := \int_{\mathbb{T}} g(\omega) e^{i2\pi k \omega}  d\omega \in \ell^2(\mathbb{Z}).
\end{equation*}
\end{definition}
\begin{example}\label{example:hatg}
   Let $g \in \ell^2(\mathbb{Z})$ be defined by $$g(n)=\begin{cases}
    1,&n=0\\
    \frac{2\sin(n\pi/2)}{n\pi},&\textnormal{otherwise}
\end{cases},$$ then we have
$$\hat{g}(\omega)=\begin{cases}
    0,&\omega\in[-1/2,-1/4)\cup(1/4,1/2]\\
    1/2,&\omega=\pm 1/4\\
    1,&\omega\in(-1/4,1/4)
\end{cases}.$$
\end{example}
The cyclic group $\mathbb{T}$ is identified with either $[0,1)$ or $[-1/2,1/2)$, as context dictates. 
The inner product on $L^2(\mathbb{T})$ is defined as:
\begin{equation*}
\langle f, g \rangle := \int_{\mathbb{T}} f(\omega) \overline{g(\omega)} d\omega, \textnormal{ for }f,g\in L^2(\mathbb{T}).
\end{equation*}

%Our discussion will revolve around three principal separable complex Hilbert spaces: $L^2(\mathbb{R})$, $\ell^2(\mathbb{Z})$, and $L^2(\mathbb{T})$, each endowed with its respective norm: $|\cdot|{L^2(\mathbb{R})}$, $|\cdot|{\ell^2(\mathbb{Z})}$, and $|\cdot|{L^2(\mathbb{T})}$. By Plancherel's theorem and our setup definitions, it follows that $|f|{\ell^2(\mathbb{Z})} = |\hat{f}|{L^2(\mathbb{T})}$ and $|f|{L^2(\mathbb{R})} = |\hat{f}|_{L^2(\mathbb{R})}$.

% \subsection{Banach densities for a set of integers} \label{sec:Banach}

\begin{definition}[Banach density]\label{def:Banach_density}
Let $\Lambda\subset\mathbb{Z}$.  
\begin{align}
    \label{upBdens} \bar{d}(\Lambda) &:=
    \limsup_{l\to\infty}\sup_{K}\frac{|\{\lambda\in [K-l,K+l]\cap\Lambda\}|}{2l}\\
    \label{lowBdens}\text{ and }~ \underline{d}(\Lambda) &:=
    \liminf_{l\to\infty}\inf_{K}\frac{|\{\lambda\in [K-l,K+l]\cap\Lambda\}|}{2l}
\end{align}
are called the upper and lower Banach densities of $\Lambda$, respectively. Moreover, if $\bar{d}(\Lambda)=\underline{d}(\Lambda)$, we say that $\Lambda$ has a Banach density $d(\Lambda):=\bar{d}(\Lambda)=\underline{d}(\Lambda)$. 
\end{definition}
For example, if $\Lambda=m\mathbb{Z}$, then $d(\Lambda)=1/m$, and if $|\Lambda|<\infty$, then $d(\Lambda)=0$.
\begin{definition}[Frame set]
A frame $\{\phi_n\}_{n\in\mathbb{Z}}$ in a separable Hilbert space $\mathcal{H}$ is a sequence of vectors satisfying the frame condition: 
\begin{equation}\label{eqn:frame}
    c\|f\|^2\leq \sum_{n\in\mathbb{Z}}|\langle f,\phi_n\rangle\|^2\leq C\|f\|^2, \textnormal{ for all }f\in\mathcal{H},
\end{equation}
for some positive constants $c$, $C>0$. The constants $c$ and $C$ are called the lower and upper frame bounds respectively, or the frame constants.  
\end{definition}
\subsection{Preliminaries for operators}
The complex vector space of all bounded linear operators on a generic separable complex Hilbert space $\mathcal{H}$ is denoted by $\mathcal{L}(\mathcal{H})$, and the set of all bounded normal linear operators on $\mathcal{H}$ is denoted by $\mathcal{N}(\mathcal{H})$, defined as: 
\begin{equation*}
    \mathcal{N}(\mathcal{H}):=\{B\in\mathcal{L}(\mathcal{H}): BB^{*}=B^{*}B\}. 
\end{equation*}
An example of such a bounded normal linear operator is $A:\ell^2(\mathbb{Z})\to\ell^2(\mathbb{Z})$ in \eqref{def:convop}, and its adjoint $A^{*}:\ell^2(\mathbb{Z})\to\ell^2(\mathbb{Z})$, which is another convolution operator with kernel $\bar{a}\in \ell^1(\mathbb{Z})$.
\begin{definition} {\it (Operator norm)}
   Let $\mathcal{H}$  be a Hilbert   space  equipped with an inner product $\langle\cdot,\cdot\rangle$  and the corresponding norm $\|\cdot\|$. The operator norm for  $B\in\mathcal{L}(\mathcal{H})$ is defined as
\begin{equation*} 
    \|B\|_{op} := \sup_{0<\|f\|\leq 1; f\in\mathcal{H}}\|B(f)\|.
\end{equation*} 
\end{definition}
If  $A$ satisfies \eqref{def:convop}, we also have \cite{davidson1996c}
\begin{equation*} 
    \|A\|_{op}=\sup\{|\lambda|: \lambda\in\spec(A)\}=\|A^{*}\|_{op},
\end{equation*}
where $\spec(A):=\{\lambda\in\mathbb{C}:\lambda\cdot\textnormal{id}-A \textnormal{ is not invertible}\}$ denotes the spectrum  of  $A$ with $\textnormal{id}$ being the identity operator.  
Similarly, if $B: L^2(\mathbb{R})\to L^2(\mathbb{R})$ is a convolution operator with kernel $b\in L^1(\mathbb{R})$, then we interpret $B^{t}: L^2(\mathbb{R})\to L^2(\mathbb{R})$, $t\geq 0$, as
\begin{equation*} 
    \mathcal{F}(B^{t}f):=(\mathcal{F}b)^{t}(\mathcal{F}f)\in L^2(\mathbb{R}).
\end{equation*}
To ensure that $\hat{a}^{t}$ and  $(\mathcal{F}b)^{t}$ are well-defined, we assume additionally that 
\begin{equation*}
    \range(\hat{a}), \range(\mathcal{F}b) \subset \mathbb{C}-\{x+i\cdot 0: x<0\}=:R,
\end{equation*}
since if $\lambda\in R$ and $t\geq 0$, then $\lambda^{t} = \exp(t \cdot\mathrm{Log}(\lambda))$ where $\mathrm{Log}: R\to\mathbb{C}$ is the principal logarithm function \cite{gamelin2003complex}. 

%\subsection{Functions on high dimensional tori} \label{sec:hdt} 
\begin{definition}[Norm of function on high dimensional tori]\label{sec:hdt}
    Let $I\subset\mathbb{T}:=[0,1)$ be a set of positive measure. Let ${\bf z}=(z_1,\cdots,z_{m})\in (L^2(I))^{m}$, i.e., $z_{j}\in L^2(I)$ for $j=1,\cdots,m$.  The norm of ${\bf z}$ is defined as 
\begin{equation} \label{def:2norm}
    \|{\bf z}\|^2:=\sum_{j=1}^{m}\|z_{j}\|_{L^2(\mathbb{T})}^2.
\end{equation}

\end{definition}

We now present two existing results that will be useful in our analysis. The first is a proposition that imposes a restriction on the spectral topology of a normal operator.

\begin{proposition}[\cite{aldroubi2017dynamical1}] \label{prop:spectrum}  If $B\in\mathcal{N}(\mathcal{H})$  with $\textnormal{dim}(\mathcal{H})=\infty$, $\mathcal{G}\subset\mathcal{H}$ with $|\mathcal{G}|<\infty$, and the system $\{B^{s}g: g\in\mathcal{G}, s\in\mathbb{N}_0\}$ satisfies the lower frame bound \footnote{The proof of this proposition, \cite[Corollary 1]{aldroubi2017dynamical1}, only uses the the lower frame bound property.}, then $B$ is unitarily equivalent to an operator
\begin{equation*}
    \sum_{j}\lambda_{j}P_{j} \quad\text{ where }\quad |\lambda_{j}|<1,
\end{equation*}
and $P_{j}$ are orthogonal projections such that $\dim(P_{j})\leq |\mathcal{G}|$.
\end{proposition}

The second result is a lemma on the operator norm on a high-dimensional torus (see \Cref{sec:hdt}).

\begin{lemma}[\cite{halmos2017introduction}] \label{lem:halmos} 
Suppose $\mathcal{A}: (L^2(\mathbb{T}))^{l}\to (L^2(\mathbb{T}))^{n}$ is defined by $(\mathcal{A}x)(\omega)=M(\omega)x(\omega)$, where the map $\omega\mapsto M(\omega)$ from $\mathbb{T}$ to the space of $n\times l$ matrices $\mathcal{M}^{nl}$ is measurable. Then \begin{equation*}
    \|\mathcal{A}\|_{op}=\textnormal{ess}\sup_{\omega\in\mathbb{T}}\|M(\omega)\|_{op}.
\end{equation*}
\end{lemma}

\section{Cardinality characterization on $\mathcal{G}$ for frame conditions} \label{sec:finitecard}
In this section, our primary aim is to conduct a rigorous analysis of the cardinality of the set $\mathcal{G}$ to verify the frame conditions stated in \eqref{eqn:sensing_set}. Our specific objective is to establish that the set $\{A^sg:g\in\mathcal{G},s\in \mathbb{N}_0\}$ cannot be  a frame system for $\ell^2(\mathbb{Z})$ when the cardinality $|\mathcal{G}|$ is finite, irrespective of the properties of the convolution kernel $a$.
\begin{theorem} \label{thm:finite}Let $A$ be defined as in \eqref{def:convop} and consider $\mathcal{G} \subset \ell^2(\mathbb{Z})$. If  the set $\{A^s g : g \in \mathcal{G}, s \in \mathbb{N}_0\}$ satisfies the lower frame bound (see \eqref{eqn:frame}):
\begin{equation} \label{prop1a}
\sum_{g \in \mathcal{G}}\sum_{s \in \mathbb{N}_0} |\langle A^{s}g, f\rangle |^2 \geq c\|f\|^2_{\ell^2(\mathbb{Z})}
\end{equation}
with a positive constant $c$, then the cardinality of $\mathcal{G}$ must be infinite.
\end{theorem}

\begin{proof}
% Since $A\in\mathcal{N}(\ell^2(\mathbb{Z}))$, it automatically inherits all the framability results for those $B\in\mathcal{N}(\ell^2(\mathbb{Z}))$ in  \Cref{prop:spectrum} with $\mathcal{H}=\ell^2(\mathbb{Z})$ and $B=A$. 
  
Since the Fourier transformation on $\ell^2(\mathbb{Z})$ is a unitary transformation from $\ell^2(\mathbb{Z})$ to $L^2(\mathbb{T})$ and $\mathcal{F}_{\mathbb{Z}}(A(f))=\hat{a}\hat{f}=:M_{\hat{a}}\hat{f}$, one has
\begin{equation} \label{specfac}
    \spec(A)=\spec(M_{\hat{a}})=\range(\hat{a}).
\end{equation}
Suppose that $|\mathcal{G}|<\infty$. Since $A\in\mathcal{N}(\ell^2(\mathbb{Z}))$ and $\{A^sg:g\in\mathcal{G},s\in\mathbb{N}_0\}$ satisfies the lower frame bound,   according to  \Cref{prop:spectrum} by setting $\mathcal{H}=\ell^2(\mathbb{Z})$ and $B=A$, $\range(\hat{a})$ must be a countable set,  and
\begin{equation} \label{hatb}
    \hat{a}(\xi)=\sum_{j}\lambda_{j}\chi_{E_{j}}(\xi),
\end{equation}
where $|\lambda_{j}|< 1$ and $E_{j}$'s are a.e. disjoint Borel subsets of $\mathbb{T}$. 
From \eqref{hatb}, one can say that  $A$ is unitarily equivalent to a diagonal operator $$\mathcal{F}_{\mathbb{Z}}A\mathcal{F}_{\mathbb{Z}}^{-1} = M_{\hat{a}} = \sum_{j}\lambda_{j}P_{j}$$  through the Fourier transform with
 $P_{j}f:=f\chi_{E_{j}}$ for $f\in L^2(\mathbb{T})$. 
 Again by \Cref{prop:spectrum}, we also have $\dim(P_{j})\leq |\mathcal{G}|<\infty$. Since the measure of $\mathbb{T}$ equals $1$,   thus some of $E_j$ has positive measure. Hence we have $\dim(L^2(E_j))=\infty$ for some $j$, which contradict with that $\dim(P_j)=\dim(L^2(E_j))<\infty$ for all $j$. Therefore, we have $|\mathcal{G}|=\infty$.
\end{proof}

\begin{remark}
The lower frame bound property remains unobtainable even when we take the iterative power in \eqref{prop1a} continuously. Specifically, if
\begin{equation*}
    \sum_{g\in\mathcal{G}}\int_0^{T} |\langle A^{t}g, f\rangle_{\mathbb{Z}}|^2\,dt\geq c\|f\|_{\ell^2(\mathbb{Z})}^2  \text{ for all } f\in\ell^2(\mathbb{Z}) \text{ and }T<\infty
\end{equation*}
holds for some $c>0$, then $\sum_{g\in\mathcal{G}}\|g\|^2=\infty$. Here\footnote{For a more precise definition of a fractional power of an operator, please refer to \cite{aldroubi2019frames}.},   $\mathcal{F}_{\mathbb{Z}}(A^{t}f):=\hat{a}^{t}\hat{f}$ for $A$. in \eqref{def:convop}. A proof of this fact is a simple adaptation of that of \cite[Lemma 5.2]{aldroubi2019frames}.  
\end{remark}
Given the results above, the following corollary can be readily derived.
\begin{corollary}
    Let $A$ be defined as in \eqref{def:convop} and consider $\mathcal{G} \subset \ell^2(\mathbb{Z})$. The set $\{A^s g : g \in \mathcal{G}, s \in \mathbb{N}_0\}$ is a frame set in $\ell^2(\mathbb{Z})$, then the cardinality of $\mathcal{G}$ must be infinite.
\end{corollary}

\section{Structural sampling:  sub-lattice sampling} \label{sec:subsampling}

In  \Cref{sec:finitecard},   \Cref{thm:finite} establishes that the frame system $\{A^sg:g\in\mathcal{G},s\in\mathbb{N}_0\}$ is only possible in $\ell^2(\mathbb{Z})$ if $|\mathcal{G}|=\infty$. In this section, our focus is on sub-lattice sampling, where the sampling set  $\Omega$ is of  the following form:
\begin{equation}\label{eqn: sub-lattice}
\Omega=\{A^se_{k}: k \in \Lambda\text{ and }s\in [N]\}, 
\end{equation}
with $\Lambda = \{mj+c: j\in\mathbb{Z},c=0,1,\cdots,L-1\}$, where $m$ and $L$ are some positive integers with $m\geq 2$, and the convolution operator $A$ is defined by the kernel $a=(a(k))_{k\in\mathbb{Z}}\in\ell^1(\mathbb{Z})$.
In the context of this sub-lattice sampling setting, we present the following key result that generalizes the finite dimensional result (see Theorem 2.12 in \cite{tang2017universal}).

\begin{theorem} \label{thm:sublatsamp}
Let $A$ be defined as in \eqref{def:convop} with $a\in\ell^1(\Z)$. If $\Omega$ in \eqref{eqn: sub-lattice} is complete in $\ell^2(\Z)$, then
\begin{equation} \label{thm1conc}
    L \geq \mathfrak{N}(t,\omega) :=\left|\left\{j: \hat{a}\left(\frac{\omega+j}m\right)=t\right\}\right|,
\end{equation}
for all $t\in \range(\hat{a})$, a.e. $\omega\in [0,1)$, and $NL\geq m$. 
\end{theorem} 

\begin{proof} First of all, let's introduce these notations:  set  $\Lambda:=m\mathbb{Z}+\mathcal{E}:=m\mathbb{Z}+\{0,\cdots,L-1\}$, 
$e(t) := e^{-i2\pi t}$, and $\tau_{c}g(x):=g(x+c)$. The assumption that $\Omega$ in \eqref{eqn: sub-lattice} is complete in $\ell^2(\Z)$ implies that 
\begin{equation} \label{assump}
    \langle f, A^{s}e_{\lambda}\rangle_{\mathbb{Z}} = (A^{*s}f)(\lambda) =0  \text{ for all }\lambda\in\Lambda, s\in [N] \text{ iff }  f\equiv 0.
\end{equation}
   Applying the Fourier transform and the Poisson summation formula on \eqref{assump} with each fixed pair of $(c,s)$ , we get
\begin{equation} \label{lefty}
     \mathfrak{A}_{m}(\omega) \diag({\bf u}_{c})\hat{{\bf f}}(\omega)= m \cdot e(-\frac{c\omega}{m}){\bf y}_{c}(\omega),
\end{equation}
where
\begin{equation} \label{matrixA}
\begin{aligned}
    \mathfrak{A}_{m}(\omega) := \begin{pmatrix} 1 & 1 & \cdots & 1\\ \overline{\hat{a}(\frac{\omega}{m})} & \overline{\hat{a}(\frac{\omega + 1}{m})} & \cdots & \overline{\hat{a}(\frac{\omega + m-1}{m})}\\
    \vdots & \vdots & \ddots & \vdots\\
    \overline{\hat{a}(\frac{\omega}{m})}^{N-1} & \overline{\hat{a}(\frac{\omega + 1}{m})}^{N-1} & \cdots & \overline{\hat{a}(\frac{\omega + m-1}{m})}^{N-1}\end{pmatrix}, & {\bf u}_{c} := \begin{pmatrix} e(\frac c m)\\ \vdots \\ e(\frac {c(m-1)}m)\end{pmatrix}, \\
    {\bf y}_{c}(\omega) := \begin{pmatrix}\mathcal{F}_{\mathbb{Z}}(S_{m}\tau_{c}f)(\omega)\\ \vdots \\ \mathcal{F}_{\mathbb{Z}}(S_{m}\tau_{c}(A^*)^{N-1}f)(\omega)\end{pmatrix} ,  \text{ and }&
    \hat{{\bf f}}(\omega) := \begin{pmatrix} \hat{f}(\frac \omega m)\\ \vdots \\ \hat{f}(\frac{\omega+m-1}m)\end{pmatrix}.
    \end{aligned}
\end{equation}
%with $ S_{m}\tau_{c}(A^*)^{s}f(j)=  (A^*)^{s}f(\lambda)$. 
Therefore, \eqref{assump} is equivalent to 
\begin{equation*} 
    \mathfrak{A}_{m}(\omega) \diag({\bf u}_{c})\hat{{\bf f}}(\omega)=0 ~\forall c\in \mathcal{E}, \text{ a.e. } \omega\in\mathbb{T} 
    \text{ iff } \hat{{\bf f}}(\omega) =0 \text{ a.e. } \omega\in\mathbb{T}.
\end{equation*}
%i.e., $[(\mathfrak{A}_{m}(\omega) \diag({\bf u}_{c}))_{c\in\mathcal{E}}]^{\top}$ is left invertible for all $c\in\mathcal{E}$ and for a.e $\omega\in [0,1)$. 
This implies that $\bigcap_{c\in\mathcal{E}} \ker(\mathfrak{A}_{m}(\omega) \diag({\bf u}_{c}))=\{0\}$, which is  equivalent to 
\begin{equation} \label{requirement}
 \sum_{c\in\mathcal{E}} (\ker(\mathfrak{A}_{m}(\omega) \diag({\bf u}_{c})))^{\perp} 
    = \sum_{c\in\mathcal{E}} {\diag}^{*}({\bf u}_{c})\ker(\mathfrak{A}_{m}(\omega))^{\perp} = \mathbb{C}^{m}
\end{equation}
for a.e. $\omega\in\mathbb{T}$.
In order to prove \eqref{requirement}, set  \[S_{t,\omega}=\left\{j+1:\hat{a}\left(\frac{\omega+j}{m}\right)=t \text{ and }j\in[m]\right\}.\]
For  fixed $t,\omega$ with $S_{t,\omega}\neq \emptyset$,   $1_{S_{t,\omega}}\in\ker(\mathfrak{A}(\omega))^\perp$, where $1_{S_{t,\omega}}(i)=\begin{cases}
    0, i\not\in S_{t,\omega}\\
    1, i\in S_{t,\omega}
\end{cases}$. 

Notice that for each fixed $\omega$, only a finite number of $t$ will result in  $S_{t,\omega}\neq \emptyset$. Let us  enumerate them as $t_0,t_1,\dots,t_{l_{\omega}}$ and denote the corresponding  vectors in $\ker(\mathfrak{A}(\omega))^\perp$  as $1_{S_{t_{0},\omega}}, 1_{S_{t_{1},\omega}}, \dots, 1_{S_{t_{l_{\omega}},\omega}}$. 

One can  see that $1_{S_{t_{0},\omega}}, 1_{S_{t_{1},\omega}}, \dots, 1_{S_{t_{l_{\omega}},\omega}}$ form an orthonormal basis for $(\ker(\mathfrak{A}_{m}(\omega)))^{\perp}$. Thus, \eqref{requirement}  is equivalent to 
$$ \spn\left\{ \overline{{\bf u}}_{c}\hdmd 1_{S_{t_{l},\omega}}: c\in\mathcal{E}, l=0,\cdots, l_{\omega}\right\}=\mathbb{C}^{m},$$   which is further equivalent  to $\oplus_{l=0}^{l_{\omega}} R_{t_{l}}=\mathbb{C}^m$ with 
$R_{t_{l}}=\spn\{\overline{{\bf u}}_{c}\hdmd 1_{S_{t_{l},\omega}}: c\in\mathcal{E}\}$, where $\oplus$ and $\hdmd$ denote the direct summation and  the Hadamard product, respectively. It is observed that
 $R_{t_{l}}$ can also be represented as $R_{t_{l}}:=\spn\{e_{j}: j\in S_{t_l,\omega}\}$, which are orthogonal spaces that compose $\mathbb{C}^{m}$. This means, in particular,
\begin{equation} \label{conc}
    |\mathcal{E}|=L\geq \left|\left\{j: \hat{a}\left(\frac {\omega+j}m\right)=t\right\}\right|=\mathfrak{N}(t,\omega),
\end{equation}
from which \eqref{thm1conc} follows.   To make sure that $\begin{pmatrix}
    \mathfrak{A}_{m}(\omega) \diag({\bf u}_{0})\\
    \vdots\\
     \mathfrak{A}_{m}(\omega) \diag({\bf u}_{L-1})
\end{pmatrix}\in\mathbb{C}^{NL\times m}$ is left invertible, then $NL\geq m$ holds.  
\end{proof}
 \begin{remark}
 ~~~
 \begin{itemize}
     \item Notice that if $\mathfrak{N}(t,\omega)>1$ for some $t\in \range(\hat{a})$ and every $\omega$ on a subset of $\mathbb{T}$ with positive measure, then the completeness of $\Omega$ requires $L>1$.
     \item In the context of sub-lattice sampling, we impose a restriction on the sampling time set, limiting it to a finite set with $N=m$. This restriction is based on the observation that the left invertibility of $\mathfrak{A}_{m}(\omega)$ relies solely on the distinctness of $\hat{a}(\frac{\omega+j}{m})$ for $j=0,\cdots,m-1$. 
     \item In \cite{davis2014dynamical}, Davis provided a sufficient condition for sub-lattice sampling, emphasizing the importance of the number of $\omega \in \mathbb{T}$ that make $\mathfrak{A}_m(\omega)$ (defined in \eqref{lefty}) singular. In contrast, our result imposes no restrictions on the number of singular points $\omega \in \mathbb{T}$. Instead, we focus on the condition that, for each fixed $\omega \in \mathbb{T}$, we examine how many $j$s can make $\hat{a}(\frac{\omega+j}{m})$ equal, though our characterization is only a necessary condition. % For the reader's convenience, we also present the result from \cite{davis2014dynamical} below.
     % \begin{theorem}
     %     Let $m\in\mathbb{Z}^+$ be fixed and let $A$ be defined as in \eqref{def:convop} with $a\in\ell^{1}(\Z)$. Suppose that $\mathfrak{A}_m(\xi)$ defined in \eqref{lefty} is singular only when $\xi\in\{\xi_i\}_{i\in I}$ with $|I|<\infty$. Let $n$ be a positive integer such that $|\xi_i-\xi_j|\neq k/n$ for any $i,j\in I$ and $k\in\{1,\cdots,n-1\}$. Then the extra samples given by $\{(S_{mn}\tau_c)f\}_{c\in\{1,\cdots,m-1\}}$ provide enough additional information to stably recover any $f\in\ell^{2}(\Z)$. 
     % \end{theorem}
 \end{itemize}
 \end{remark}
Based on  \Cref{thm:sublatsamp}, one can easily get the following corollary for the case of
\begin{equation}\label{eqn:sub-lattice0} \Omega_0=\left\{A^se_{mj}: j\in\mathbb{Z} \text{ and }s\in [N]\right\}.
\end{equation}

\begin{corollary} \label{cor:sublatsamp}
Let $A$ be a convolution operator given by  %\eqref{def:convop} with 
the  kernel  $a\in\ell^1(\Z)$. Suppose that $N\geq m$.  Then  $\Omega_0$ 
  is complete in $\ell^2(\Z)$ only if $\mathfrak{N}(t,\omega)=1$ for all $t\in \range(\hat{a})$, a.e. $\omega\in \mathbb{T}$.%, which implies that $\hat{a}$ is monotone.
\end{corollary} 
Additional insights can be derived from \Cref{cor:sublatsamp}, providing further understanding. Specifically, \Cref{prop:twosamp} will demonstrate that if  $a\in\ell^1(\Z)$ and $N=m$, the completeness of $\Omega_0$ in \eqref{eqn:sub-lattice0} within $\ell^2(\Z)$ is equivalent to it being a frame system in $\ell^2(\Z)$. To establish this proposition, we will first require the following lemma, which asserts the stability of the inversion process in \eqref{lefty} with $c=0$, given the existence of the inverse matrix $\mathfrak{A}_{m}^{-1}(\omega)$ for every $\omega\in \mathbb{T}$.

\begin{lemma} \label{lem:ptwbd}
 Suppose that $N=m$   and  $c=0$ in \eqref{lefty}.  
Let $\mathfrak{A}^{-1}_{m}: (L^2(\mathbb{T}))^{m}\to (L^2(\mathbb{T}))^{m}$ be a function such that
\begin{equation} \label{inverse}
    (\mathfrak{A}^{-1}_{m}{\bf y})(\omega):={\bf \hat{f}}(\omega)/m.
\end{equation}
where ${\bf \hat{f}}$, ${\bf y}$ are defined in \eqref{lefty}. Then $\mathfrak{A}^{-1}_{m}$ in \eqref{inverse} is bounded iff $\mathfrak{A}_{m}(\omega)$ in \eqref{matrixA} is invertible for every $\omega\in \mathbb{T}$.
\end{lemma}
\begin{proof} 
The only if direction is obvious; hence we focus on the if direction. Suppose $\mathfrak{A}^{-1}_{m}(\omega)$ exists for every $\omega\in \mathbb{T}$. To see if this leads to $\mathfrak{A}^{-1}_{m}$ in \eqref{inverse} being a bounded operator, we apply \Cref{lem:halmos} to our context and obtain
\begin{equation} \label{invnorm}
    \|\mathfrak{A}^{-1}_{m}\|_{op}=\sup_{\omega\in \mathbb{T}}\|\mathfrak{A}^{-1}_{m}(\omega)\|_{op},
\end{equation}
since $\hat{a}\in C(\mathbb{T})$. Furthermore, for every $\omega\in \mathbb{T}$, the following holds \cite{gautschi1978inverses}, 
\begin{equation} \label{invAnorm}
    \|\mathfrak{A}^{-1}_{m}(\omega)\|_{op}\leq\sqrt{m}\max_{0\leq i\leq m-1}\prod_{j\not=i,j=0}^{m-1}\frac{1+|\hat{a}(\frac{\omega+j}{m})|}{|\hat{a}(\frac{\omega+j}{m})-\hat{a}(\frac{\omega+i}{m})|}.
\end{equation}
Combining \eqref{invnorm} and \eqref{invAnorm}, it yields
\begin{align}
    \nonumber &\|\mathfrak{A}^{-1}_{m}\|_{op} \leq\sqrt{m}\max_{0\leq i\leq m-1} \sup_{\omega\in \mathbb{T}}\prod_{j\not=i,j=0}^{m-1}\frac{1+|\hat{a}(\frac{\omega+j}{m})|}{|\hat{a}(\frac{\omega+j}{m})-\hat{a}(\frac{\omega+i}{m})|}\\
    \label{keyeq} &\leq\sqrt{m}(1+\max_{\omega\in\mathbb{T}}|\hat{a}(\omega)|)^{m-1} \max_{0\leq i\leq m-1} \sup_{\omega\in \mathbb{T}}\prod_{j\not=i,j=0}^{m-1}\frac{1}{|\hat{a}(\frac{\omega+j}{m})-\hat{a}(\frac{\omega+i}{m})|}.
\end{align}
Due to the continuity of $\hat{a}$ and the invertibility of $\mathfrak{A}_{m}(\omega)$, we have that for $i\neq j$, 
\begin{equation} \label{min}
    \inf_{\omega\in \mathbb{T}}\bigg|\hat{a}\bigg(\frac{\omega+j}{m}\bigg)-\hat{a}\bigg(\frac{\omega+i}{m}\bigg)\bigg|=\min_{\omega\in \mathbb{T}}\bigg|\hat{a}\bigg(\frac{\omega+j}{m}\bigg)-\hat{a}\bigg(\frac{\omega+i}{m}\bigg)\bigg|=:\mathfrak{c}_{ij}>0.
\end{equation}
Since there are only a finite number of pairs $(i,j)$ with $i\neq j$, we can take $\min_{i\not= j}\mathfrak{c}_{ij}=:\mathfrak{c}>0$. Substituting this back to \eqref{keyeq}, we obtain that
\begin{equation*}
\|\mathfrak{A}^{-1}_{m}\|_{op}\leq\sqrt{m}(\mathfrak{c}^{-1}(1+\max_{\omega\in\mathbb{T}}|\hat{a}(\omega)|)^{m-1}<\infty.
\end{equation*}
Therefore, $\mathfrak{A}^{-1}_{m}$ is bounded.  
\end{proof}

\begin{theorem} \label{prop:twosamp}
Suppose $a\in\ell^1(\mathbb{Z})$. Then 
\begin{equation} \label{twosampconc}
    \Omega_0 \text{ is a frame system in  }\ell^2(\Z) \text{iff }   \Omega_0 \text{ is complete in }\ell^2(\Z).
\end{equation}
\end{theorem} 
\begin{proof}
It is obvious that when $\Omega_0$ is a frame system in $\ell^2(\Z)$, then $\Omega_0$ is complete in $\ell^2(\Z)$.  Next, we only need to show that   $\Omega_0$ being complete in $\ell^2(\Z)$ can imply that $\Omega_0$ is a frame system in $\ell^2(\Z)$. 

It is worth noting that it suffices to establish the conclusion for the case where $N=m$, given the properties of the Vandermonde matrix. Therefore, our focus will  be on proving the case that $N=m $  for $\Omega_0$.  It follows from Theorem \ref{thm:sublatsamp} that $\Omega_0$ is complete in $\ell^2(\Z)$ only if the system \eqref{lefty} is left invertible for every $\omega\in \mathbb{T}$; moreover, \eqref{thm1conc} becomes 
\begin{equation} \label{contL}
    L \geq \sup_{t\in \textnormal{range}(\hat{a})} \mathfrak{N}(t,\omega).
\end{equation}
If $\mathfrak{A}_{m}(\omega)$ in \eqref{matrixA} is not invertible for some $\omega\in \mathbb{T}$,  then $\mathfrak{N}(t,\omega)>1$ for some $t\in \textnormal{range}(\hat{a})$, which in turn implies that ${\bf \hat{f}}(\omega)$ is not unique for some $\omega$.  
\\
Let ${\bf \hat{f}}$ be as in  \eqref{matrixA}.  Recall from \eqref{def:2norm} (with $I=[0,1/m)$) that 
\begin{equation*}
    \|{\bf \hat{f}}\|^2= \sum_{j=0}^{m-1}\int_{j/m}^{(1+j)/m} |\hat{f}(\omega)|^2\,d\omega=\|\hat{f}\|^2_{L^2(\mathbb{T})}=\|f\|^2_{\ell^2(\mathbb{Z})}.
\end{equation*}
$\mathfrak{A}^{-1}_{m}$ being a bounded operator implies that  
% \begin{equation}\label{eqn:LFB}
%  \sum_{j\in\mathbb{Z}}\sum_{s=0}^{m-1}|\langle A^se_{mj},f\rangle|^2\geq \frac{1}{m^2\|\mathfrak{A}^{-1}_{m}\|^2_{op}}\|f\|^2_{\ell^2(\mathbb{Z})}.
% \end{equation}
%% Here is the detailed explaination for the above inequality
\begin{equation*} 
\begin{aligned}
    \|f\|^2_{\ell^2(\mathbb{Z})}=&\|{\bf \hat{f}}\|^2=\|m\mathfrak{A}^{-1}_m{\bf y}_0\|^2\leq m^2\|\mathfrak{A}^{-1}_{m}\|^2_{op}\|{\bf y}_0\|^2 \\
=&m^2\|\mathfrak{A}^{-1}_{m}\|^2_{op}\sum_{s=0}^{m-1}\|\mathcal{F}_{\mathbb{Z}}(S_{m}A^{*s}f)\|^2_{L^2(\mathbb{T})}\\
=&m^2\|\mathfrak{A}^{-1}_{m}\|^2_{op}\sum_{s=0}^{m-1}\|S_{m}A^{*s}f\|^2_{\ell^2(\mathbb{Z})}\\
=&m^2\|\mathfrak{A}^{-1}_{m}\|^2_{op}\sum_{j\in\mathbb{Z}}\sum_{s=0}^{m-1}|(A^{*s}f)(mj)|^2\\
=&m^2\|\mathfrak{A}^{-1}_{m}\|^2_{op}\sum_{j\in\mathbb{Z}}\sum_{s=0}^{m-1}|\langle A^se_{mj},f\rangle|^2,
\end{aligned} 
\end{equation*}
i.e., \begin{equation}\label{eqn:LFB}
 \sum_{j\in\mathbb{Z}}\sum_{s=0}^{m-1}|\langle A^se_{mj},f\rangle|^2\geq \frac{1}{m^2\|\mathfrak{A}^{-1}_{m}\|^2_{op}}\|f\|^2_{\ell^2(\mathbb{Z})}.
\end{equation}
In addition,
\begin{equation}\label{eqn:UFB}
    \begin{aligned}
 \sum_{j\in\mathbb{Z}}\sum_{s=0}^{m-1}|\langle A^se_{mj},f\rangle|^2&=\sum_{s=}^{m-1}\|S_{m}A^{*s}f\|_{\ell^2(\mathbb{Z})}^2 \\
 =&\sum_{s=}^{m-1}\|\mathcal{F}_{\mathbb{Z}}(S_{m}A^{*s}f)\|_{L^2(\mathbb{T})}^2=\|{\bf y}_0\|^2=\frac{1}{m^2}\|\mathfrak{A}_{m}\hat{{\bf f}}\|^2\\
 \leq& \frac{1}{m^2}\|\mathfrak{A}_m\|_{op}^2\|\hat{{\bf f}}\|^2
 =\frac{1}{m^2}\|\mathfrak{A}_m\|_{op}^2\|f\|^2_{\ell^2(\mathbb{Z})}.
    \end{aligned}
\end{equation}
Therefore, $\Omega_0$ is a frame system in $\ell^2(\Z)$.  
\end{proof}
 
\section{Characterization of sensor density, maximum spatial gaps, and frame bounds} \label{sec:density}
In the previous section, we discussed   that  a stable sampling  set not only   satisfies the cardinality requirement (\Cref{sec:finitecard}) but also  depends on the convolution kernel. In this section, we introduce additional factors that characterize a stable sampling set.  Let's start by examining an illustrative example.

\begin{example}
Let $\hat{a}\in C^1(\mathbb{T})$ be one-to-one. Then   the set $\{(A^{s}f)(mj): j\in\mathbb{Z}, s\in [m]\}$ is always stable. Indeed,   the lower frame bound  can be  $\frac{1}{m^2}\|\mathfrak{A}^{-1}_{m}\|_{op}^{-2}$ with  
\begin{equation*} \label{exA} \|\mathfrak{A}^{-1}_{m}\|_{op}^{-1}=\inf_{\omega\in \mathbb{T}} \|\mathfrak{A}^{-1}_{m}(\omega)\|_{op}^{-1}   \geq\frac{1}{\sqrt{m}}\min_{0\leq i\leq m-1}\inf_{\omega\in\mathbb{T}}\prod_{l\not=i,l=0}^{m-1}\frac{|\hat{a}(\frac{\omega+l}{m})-\hat{a}(\frac{\omega+i}{m})|}{1+|\hat{a}(\frac{\omega+l}{m})|}.
\end{equation*}
The factor $|\hat{a}(\frac{\omega+l}{m})-\hat{a}(\frac{\omega+i}{m})|$ is bounded below by $\frac{\kappa}{m}$, where $\kappa$ satisfies
$$\inf_{\omega\in\mathbb{T}}\left|\frac{d}{d\omega}\hat{a}(\omega)\right|\geq\kappa>0.$$
Since the Banach density (see \Cref{def:Banach_density}) of $m\mathbb{Z}$ is $1/m$, one might be tempted to consider the naive strategy of increasing $m$, in order to achieve a high placement cost effectiveness. However, one might as well decrease the sampling's noise tolerance in doing so. More precisely, when $f$ is corrupted by a Gaussian additive noise $N(0,\sigma^2)$ - and observed as $\tilde{f}$ - then the reconstruction error, $\|f-\tilde{f}\|_{\ell^2(\mathbb{Z})}$, can be as large as $\Omega(\sigma m^{-5/2}\|\mathfrak{A}^{-1}\|_{op})$ \cite{aldroubi2015exact}. Hence, if $\kappa$ is small and $m$ grows large, the lower bound in \eqref{exA} is significantly reduced. This lessens the chance of stable reconstruction in the presence of noise.
\end{example}
The above example highlights the influence of sampling density on stability and emphasizes the importance of kernel adaptivity for efficient sampling. In this section, we explore how sampling density varies with respect to the kernel profile $a$ when a stable sampling is provided (or the frame lower and upper bounds are given). We set up our experiment as follows: we consider $\mathbb{T}:=[-1/2,1/2)$ and assume that $\hat{a}\in C^1(\mathbb{T})$ satisfies the conditions 
% \begin{equation} \label{growth}
%     \hat{a}(\omega)\geq 0 \quad\text{ and }\quad \mathcal{K}\geq \frac{d \hat{a}}{d\omega}(\omega)\geq\kappa>0, \quad\text{ for }\quad\omega\in \mathbb{T}.
% \end{equation}
\begin{equation} \label{growth}
   \mu\geq \hat{a}(\omega)\geq\nu> 0 \quad\text{ and }\quad \left|\frac{d \hat{a}}{d\omega}(\omega)\right|\leq \mathcal{K}<\infty, \quad\text{ for }\quad\omega\in \mathbb{T}.
\end{equation}

 \begin{theorem}\label{prop:density} 
Let the convolution operator $A$ defined by the kernel $a$ satisfy the properties in \eqref{growth}. Suppose that there are two positive constants $c_{max}$ and $c_{min}$ such that for any $f\in \ell^2(\mathbb{Z})$, the samples \begin{equation} \label{typsampset}
    \{A^{s}f(\lambda): \lambda\in\Lambda \subset \mathbb{Z}, s\in [N]\},
\end{equation} satisfy the frame inequality
\begin{equation} \label{framebds}
c_{min}\|f\|_{\ell^2(\mathbb{Z})}^2\leq \sum_{s=0}^{N-1}\sum_{\lambda\in\Lambda} |A^{s}f(\lambda)|^2\leq c_{max}\|f\|_{\ell^2(\mathbb{Z})}^2.
\end{equation}
 Then there exist two constants $c_{a}, C_{a}$, only depending  on the kernel $a$, such that
\begin{equation} \label{densities}
    \underline{d}(\Lambda) \geq \max\left\{\frac{c_{a}c_{min}}{2c_{max}C_{a}},\frac{c_{\min}}{3C_a}\right\}
    \text{ and }\quad
    \bar{d}(\Lambda)\leq \min\left\{\frac{c_{max}}{c_{a}}, \frac{3}{2}\right\}. 
\end{equation}
\end{theorem} 
Before presenting the proof for  \Cref{prop:density}, we will first introduce and establish a supporting lemma.

\begin{lemma} \label{lem:requisite}
Let $A$ be a convolution operator in $\ell^2(\mathbb{Z})$ defined by the kernel $a$ and assume that $a$ satisfies \eqref{growth}. Then there exist $C_{a}>c_{a}>0$ such that for each $N\in\mathbb{N}$, 
\begin{align}
    &\sum_{s=0}^{N-1} |A^{s}(g)(\lambda)|^2 \leq\frac{C_{a}}{1+\lambda^2}, \text{ for all }\lambda\in\mathbb{Z}\label{eqn:upperbnd}\\
  \text{and } &\sum_{s=0}^{N-1} |A^{s}(g))(\lambda)|^2 \geq c_{a}, \quad\text{ for }\lambda=0,\pm 1\label{eqn:lowerbnd}, 
\end{align}
where 
$g \in \ell^2(\mathbb{Z})$ is defined by $g(n)=\begin{cases}
    1,&n=0\\
    \frac{2\sin(n\pi/2)}{n\pi},&\textnormal{otherwise}.
\end{cases}$
\end{lemma}

\begin{proof}
Note that when $s=0$, we have that
\begin{equation}\label{eqn:bnd4sinc}
    |g(\lambda)|^2=\begin{cases}1,&\lambda=0\\
    \frac{4\sin^2(\lambda\pi/2)}{\lambda^2\pi^2},&\textnormal{otherwise}
    \end{cases}\leq 1/(1+\lambda^2).
\end{equation} When $s\geq 1$, note that the Fourier transformation of $g$ is $$\hat g(\omega)=\begin{cases}
    0,&\omega\in [-1/2,-1/4)\cup(1/4,1/2]\\
    1/2,&\omega=\pm 1/4\\
    1,&\omega\in(-1/4,1/4)
\end{cases}.$$  Using integration by parts, we obtain that 
\begin{align*}
    &(A^{s}g)(n)\\
    =&(a^{(s)}\ast g) (n) = \int_{-1/4}^{1/4} \hat{a}^{s}(\omega)e^{i2\pi n\omega}\,d\omega \\
    =& \frac{\hat{a}(1/4)^{s} e^{i\pi n/2}-\hat{a}(-1/4)^{s} e^{-i\pi n/2}}{i2\pi n} - \frac{s}{i2\pi n}\int_{-1/4}^{1/4}e^{i2\pi n\omega}\hat{a}^{s-1}(\omega)\frac{d \hat{a}}{d\omega}(\omega)\,d\omega.
\end{align*}
 Therefore, by the triangle inequality,
\begin{equation*}
    |n\cdot (a^{(s)}\ast g) (n)| \leq\frac{\mu^{s}}{\pi} + \frac{s\mu^{s-1}\mathcal{K}}{2\pi},
\end{equation*}
and hence
\begin{equation}\label{eqn:bnd4xf}
\begin{aligned}
    n^2\sum_{s=1}^{N-1}|(a^{(s)}\ast g)(n)|^2\leq& \sum_{s=1}^{N-1} \bigg(\frac{\mu^{s}}{\pi} + \frac{s\mu^{s-1}\mathcal{K}}{2\pi}\bigg)^2\\
    \leq& \frac{\mu^{2(N-1)}-1}{\mu^2-1}\cdot\bigg(\frac{2\mu^2}{\pi^2}+ \frac{(N-1)^2\mathcal{K}^2}{2\pi^2}\bigg).
    \end{aligned}
\end{equation} 
In addition, \begin{align}\label{eqn:bnd4f}
   |(a^{(s)}\ast g) (n)| &= \left|\int_{-1/4}^{1/4} \hat{a}^{s}(\omega)e^{i2\pi n\omega}\,d\omega \right|\leq \int_{-1/4}^{1/4}\left|\hat{a}^{s}(\omega)e^{i2\pi n\omega}\,\right|d\omega \leq \mu^s/2.
\end{align}
 Combining \eqref{eqn:bnd4sinc}, \eqref{eqn:bnd4xf},  and \eqref{eqn:bnd4f}, we can derive  \eqref{eqn:upperbnd}.

 For \eqref{eqn:lowerbnd}, %note that if $\omega\in [-1/2+\eta,1/2)$ then from \eqref{growth}, $\hat{a}(\omega)\geq\nu+\kappa\eta$. Hence for $|n|\leq 1$,
 note that for $|n|\leq 1$,
\begin{align*}
    |(a^{(s)}\ast g)(n)| &=\left|2\int_{-1/4}^{1/4}\hat{a}^{s}(\omega)e^{i 2\pi n\omega}\,d\omega\right|\geq\left|2\int_{-1/4}^{1/4}\hat{a}^{s}(\omega)\cos(2\pi n\omega)\,d\omega\right|.
    \end{align*}
Additionally,  note that when $|n|\leq 1$, we have that $\cos(2\pi n \omega)\geq 0$ for all $\omega\in[-1/4,1/4]$. Therefore, 
    \begin{align*}
 |(a^{(s)}\ast g)(n)|\geq &\left|2\nu^{s}\int_{-1/4}^{1/4}\cos(2\pi n\omega)\,d\omega\right|=\left|2\nu^{s}\frac{\sin(n\pi/2)}{n\pi}\right|\geq\frac{2}{\pi}\nu^{s}.%\frac{16}{3\sqrt{2}\pi}(\nu+\kappa/4)^{s}
    \end{align*}
    Therefore,  we have
    $ \sum_{s=0}^{N-1} |(a^{(s)}\ast g)(n)|^2  \geq  \frac{4}{\pi^2}\cdot\frac{\nu^{2(N-1)}-1}{\nu^2-1}$.
\end{proof}
We are now ready to provide the proof of \Cref{prop:density}. The main idea is to effectively utilize the frame inequality in \eqref{framebds} for the specific $g$ described in \Cref{lem:requisite}.
\begin{proof}[The proof of \Cref{prop:density}]
According to \Cref{lem:requisite}, the following two inequality holds  with  constants $c_{a}, C_{a}$: 
\begin{align}
    \label{upp1} &\sum_{s=0}^{N-1} |A^{s}(g)(\lambda)|^2 \leq\frac{C_{a}}{1+\lambda^2},\forall \lambda\in\mathbb{Z}\\
    \label{low1} &\sum_{s=0}^{N-1} |A^{s}(g)(\lambda)|^2 \geq c_{a} \quad\text{ for }\lambda=0,\pm 1, %\quad |x|\leq 3/2. 
\end{align}
where  $g(n)=\begin{cases}
    1,&n=0\\
    \frac{2\sin(n\pi/2)}{n\pi},&\textnormal{otherwise}
\end{cases}$.

Let $N(\Lambda):=\sup_{x\in\mathbb{Z}} \#(\Lambda\cap[x-1,x+1])$. 
From \eqref{low1} and  \eqref{framebds}, we have  that
\begin{align*}
    \#(\Lambda\cap[x-1,x+1])\leq&\sum_{\lambda\in\Lambda\cap[x-1,x+1]}\frac{1}{c_a}\sum_{s=0}^{N-1}|A^s(g)(\lambda-x)|^2\\
    \leq&\frac{1}{c_a}\sum_{\lambda\in\Lambda}\sum_{s=0}^{N-1}|A^s(g)(\lambda-x)|^2\\
    \leq&\frac{c_{\max}}{c_a}\left\|g(\cdot-x))\right\|_{\ell^2(\mathbb{Z})}^2=\frac{c_{\max}}{c_a}\left\|g\right\|_{\ell^2(\mathbb{Z})}^2=\frac{2c_{\max}}{c_{a}}
\end{align*}
Therefore,
\begin{equation} \label{updens}
    \bar{d}(\Lambda)\leq\frac{N(\Lambda)}{2}\leq\frac{c_{max}}{c_{a}}.
\end{equation}
In addition, $ \bar{d}(\Lambda)\leq 3/2$. Thus,  $ \bar{d}(\Lambda)\leq \min\{c_{\max}/{c_a},3/2\}$.

Suppose that there exists some $u\in\mathbb{Z}$ and let $R$ be the largest positive integer such that
$\Lambda\cap [u-R,u+R]=\emptyset$ and
\begin{equation} \label{nonempty}
    \Lambda\cap [u-R-1,u+R+1]\not=\emptyset.
\end{equation}
%Since $g\in\ell^2(\mathbb{Z})$ is translation invariant, 
By translation invariance,  we may assume $u=0$. Substitute $g\in\ell^2(\mathbb{Z})$ to \eqref{framebds}. From \eqref{upp1}, \eqref{updens}
\begin{align*}
   2c_{\min} &=c_{\min}\|g\|_{\ell^2(\mathbb{Z})}^2 \leq \sum_{s=0}^{N-1}\sum_{\lambda\in\Lambda} |g\ast a^{(s)}(\lambda)|^2\leq\sum_{\lambda\in\Lambda}\frac{C_{a}}{1+\lambda^2}\\
&=\left(\sum_{k=0}^{\infty}\sum_{\lambda\in\Lambda\cap [R+2k,R+2k+2]} + \sum_{k=0}^{\infty}\sum_{\lambda\in\Lambda\cap [-R-2k-2,-R-2k]} \right)\frac{C_{a}}{1+\lambda^2}\\
    &\leq 2N(\Lambda)C_{a}\sum_{k=0}^{\infty}\frac{1}{1+(R+2k)^2}\\
    &\leq 2C_{a}\min\left\{\frac{c_{max}}{c_{a}},3/2\right\}\int_{R}^{\infty}\frac{1}{2x^2}\,dx=\min\left\{\frac{C_{a}c_{max}}{Rc_{a}}, \frac{3C_a}{2R}\right\}.
\end{align*}
According to \eqref{nonempty}, we have
\begin{equation*}
    R\leq\min\left\{ \frac{C_{a}c_{\max}}{2c_{a}c_{\min}},\frac{3C_a}{4c_{\min}} \right\}\Rightarrow \underline{d}(\Lambda)\geq\frac{1}{2(R+1)}\geq\frac{1}{4R}\geq\max\left\{\frac{c_{a}c_{min}}{2c_{max}C_{a}},\frac{c_{\min}}{3C_a}\right\}.
\end{equation*}
which is the first part of \eqref{densities}.
\end{proof}
\begin{remark}
 ~~~
 \begin{itemize}
\item From the proof of the lower Banach density in \Cref{prop:density}, it becomes evident that if the frame (lower and upper) bounds are given, the maximal gap between the spatial samples cannot be arbitrarily large.
     \item Note that the conclusion of \Cref{prop:density} is parallel to \cite[Theorem 4.4]{aldroubi2021sampling}, which provides the density requirements for the spatial sampling set $\Lambda\subseteq \mathbb{R}$ for the Parley-Wiener space. Their result is established under the conditions that the sampling time set is an interval and the convolution kernel satisfies similar conditions to ours.  
 \end{itemize} 
\end{remark} 
 \section{Conclusion and future work}
 In this paper, we have  investigated the dynamical sampling problem on $\ell^2(\mathbb{Z})$  driven by a convolution operator.  Our study highlights the critical roles that sampling density and kernel profile play in determining reconstruction stability. Specifically, we found that a finite spatial sampling set often results in insufficient reconstruction capabilities. However, under certain conditions, we demonstrated that  stable reconstruction, as measured by frame lower and upper bounds, effectively controls the maximal gap between spatial samples. We also delved into the stability of dynamical sampling on sub-lattices  of the form $\Lambda = \{mj+c: j\in\mathbb{Z},c=0,1,\cdots,L-1\}$. Our findings elucidate the relationship between 
$L$, $m$, and the properties of the convolution kernel when the initial signal is guaranteed to be reconstructed from the given samples. Particularly, for the specific case where  $L=1$ and the convolution kernel is regular, the stability of the sampling set hinges primarily on its ability to differentiate between signals. We thus established that the dual system generated by dynamical sub-lattice sampling forms a frame for $\ell^2(\mathbb{Z})$ if and only if it is complete. 

There are two lines for future work. Firstly, we would like to generalize our characterizations of dynamical sampling in higher-dimensional settings. Secondly, we will consider various characterizations of the sampling set to ensure the reconstruction of both the initial signal and the convolution operator from the given samples. These investigations will further enhance our understanding and application of dynamical sampling in more complex scenarios. 
 \section*{Acknowledgement}
S. Tang was partially supported by Regents Junior Faculty fellowship, Faculty Early Career Development Awards sponsored by University of California Santa Barbara, Hellman Family Faculty Fellowship, NSF DMS \#2111303 and DMS Career \#2340631.  Y. Xie was supported in part by the NIH grants U01DE029255, U01DE033330,
R01HL166508, R01DE026728 and RO3DE027399 and NSF grants IOS2107215/IIS 2123260. Finally, L. Huang and S. Tang are greatly indebted to Akram Aldroubi for his mentorship and introduction to dynamical sampling problems.

\bibliography{sampling.bib} 

\begin{thebibliography}{10}

\bibitem{aceska2015multidimensional}
Roza Aceska, Armenak Petrosyan, and Sui Tang.
\newblock Multidimensional signal recovery in discrete evolution systems via
  spatiotemporal trade off.
\newblock {\em Sampling Theory in Signal and Image Processing}, 14(2):153--169,
  2015.

\bibitem{aldroubi2017dynamical}
Akram Aldroubi, Carlos Cabrelli, Ursula Molter, and Sui Tang.
\newblock Dynamical sampling.
\newblock {\em Applied and Computational Harmonic Analysis}, 42(3):378--401,
  2017.

\bibitem{aldroubi2013dynamical}
Akram Aldroubi, Jacqueline Davis, and Ilya Krishtal.
\newblock Dynamical sampling: Time--space trade-off.
\newblock {\em Applied and Computational Harmonic Analysis}, 34(3):495--503,
  2013.

\bibitem{aldroubi2015exact}
Akram Aldroubi, Jacqueline Davis, and Ilya Krishtal.
\newblock Exact reconstruction of signals in evolutionary systems via
  spatiotemporal trade-off.
\newblock {\em Journal of Fourier Analysis and Applications}, 21(1):11--31,
  2015.

\bibitem{aldroubi2021sampling}
Akram Aldroubi, Karlheinz Gr{\"o}chenig, Longxiu Huang, Philippe Jaming, Ilya
  Krishtal, and Jos{\'e}~Luis Romero.
\newblock Sampling the flow of a bandlimited function.
\newblock {\em The Journal of Geometric Analysis}, pages 1--35, 2021.

\bibitem{aldroubi2018dynamical}
Akram Aldroubi, Longxiu Huang, Ilya Krishtal, Akos Ledeczi, Roy~R Lederman, and
  Peter Volgyesi.
\newblock Dynamical sampling with additive random noise.
\newblock {\em Sampling Theory in Signal and Image Processing}, 17:153--182,
  2018.

\bibitem{aldroubi2019frames}
Akram Aldroubi, Longxiu Huang, and Armenak Petrosyan.
\newblock Frames induced by the action of continuous powers of an operator.
\newblock {\em Journal of Mathematical Analysis and Applications},
  478(2):1059--1084, 2019.

\bibitem{aldroubi2016krylov}
Akram Aldroubi and Ilya Krishtal.
\newblock Krylov subspace methods in dynamical sampling.
\newblock {\em Sampling Theory in Signal and Image Processing}, 15:9--20, 2016.

\bibitem{aldroubi2020phaseless}
Akram Aldroubi, Ilya Krishtal, and Sui Tang.
\newblock Phaseless reconstruction from space--time samples.
\newblock {\em Applied and Computational Harmonic Analysis}, 48(1):395--414,
  2020.

\bibitem{aldroubi2017dynamical1}
Akram Aldroubi and Armenak Petrosyan.
\newblock Dynamical sampling and systems from iterative actions of operators.
\newblock In {\em Frames and Other Bases in Abstract and Function Spaces},
  pages 15--26. Springer, 2017.

\bibitem{beinert2023phase}
Robert Beinert and Marzieh Hasannasab.
\newblock Phase retrieval and system identification in dynamical sampling via
  prony’s method.
\newblock {\em Advances in Computational Mathematics}, 49(4):56, 2023.

\bibitem{cabrelli2020dynamical}
Carlos Cabrelli, Ursula Molter, Victoria Paternostro, and Friedrich Philipp.
\newblock Dynamical sampling on finite index sets.
\newblock {\em Journal d'analyse math{\'e}matique}, 140(2):637--667, 2020.

\bibitem{candes2006stable}
Emmanuel~J Candes, Justin~K Romberg, and Terence Tao.
\newblock Stable signal recovery from incomplete and inaccurate measurements.
\newblock {\em Communications on Pure and Applied Mathematics: A Journal Issued
  by the Courant Institute of Mathematical Sciences}, 59(8):1207--1223, 2006.

\bibitem{cheng2021estimate}
Jiahui Cheng and Sui Tang.
\newblock Estimate the spectrum of affine dynamical systems from partial
  observations of a single trajectory data.
\newblock {\em Inverse Problems}, 38(1):015004, 2021.

\bibitem{christensen2021survey}
Ole Christensen and Marzieh Hasannasab.
\newblock A survey on frame representations via dynamical sampling.
\newblock {\em arXiv preprint arXiv:2201.00038}, 2021.

\bibitem{davidson1996c}
Kenneth~R Davidson.
\newblock {\em C*-algebras by example}, volume~6.
\newblock American Mathematical Soc., 1996.

\bibitem{davis2014dynamical}
Jacqueline Davis.
\newblock Dynamical sampling in infinite dimensions with and without a forcing
  term.
\newblock {\em AMS Contemporary Mathematics (CONM) BS}, pages 167--174, 2014.

\bibitem{gamelin2003complex}
Theodore Gamelin.
\newblock {\em Complex analysis}.
\newblock Springer Science \& Business Media, 2003.

\bibitem{gautschi1978inverses}
Walter Gautschi.
\newblock On inverses of vandermonde and confluent vandermonde matrices.
\newblock {\em Numerische Mathematik}, (4):117--123, 1962.

\bibitem{grochenig2015minimal}
Karlheinz Gr{\"o}chenig, Jos{\'e}~Luis Romero, Jayakrishnan Unnikrishnan, and
  Martin Vetterli.
\newblock On minimal trajectories for mobile sampling of bandlimited fields.
\newblock {\em Applied and Computational Harmonic Analysis}, 39(3):487--510,
  2015.

\bibitem{halmos2017introduction}
Paul~R Halmos.
\newblock {\em Introduction to Hilbert space and the theory of spectral
  multiplicity}.
\newblock Courier Dover Publications, 2017.

\bibitem{hormati2009distributed}
Ali Hormati, Olivier Roy, Yue~M Lu, and Martin Vetterli.
\newblock Distributed sampling of signals linked by sparse filtering: Theory
  and applications.
\newblock {\em IEEE Transactions on Signal Processing}, 58(3):1095--1109, 2009.

\bibitem{huang2021robust}
Longxiu Huang, Deanna Needell, and Sui Tang.
\newblock Robust recovery of bandlimited graph signals via randomized dynamical
  sampling.
\newblock {\em Information and Inference}, 2024.

\bibitem{jaye2022sufficient}
Benjamin Jaye and Mishko Mitkovski.
\newblock A sufficient condition for mobile sampling in terms of surface
  density.
\newblock {\em Applied and Computational Harmonic Analysis}, 61:57--74, 2022.

\bibitem{kummerle2022learning}
Christian K{\"u}mmerle, Mauro Maggioni, and Sui Tang.
\newblock Learning transition operators from sparse space-time samples.
\newblock {\em arXiv preprint arXiv:2212.00746}, 2022.

\bibitem{lu2011localization}
Yue~M Lu, Pier~Luigi Dragotti, and Martin Vetterli.
\newblock Localization of diffusive sources using spatiotemporal measurements.
\newblock In {\em 2011 49th Annual Allerton Conference on Communication,
  Control, and Computing (Allerton)}, pages 1072--1076. IEEE, 2011.

\bibitem{lu2009spatial}
Yue~M Lu and Martin Vetterli.
\newblock Spatial super-resolution of a diffusion field by temporal
  oversampling in sensor networks.
\newblock In {\em 2009 IEEE International Conference on Acoustics, Speech and
  Signal Processing}, pages 2249--2252. IEEE, 2009.

\bibitem{nashed2010sampling}
M~Zuhair Nashed and Qiyu Sun.
\newblock Sampling and reconstruction of signals in a reproducing kernel
  subspace of lp (rd).
\newblock {\em Journal of Functional Analysis}, 258(7):2422--2452, 2010.

\bibitem{ranieri2011sampling}
Juri Ranieri, Amina Chebira, Yue~M Lu, and Martin Vetterli.
\newblock Sampling and reconstructing diffusion fields with localized sources.
\newblock In {\em 2011 IEEE International Conference on Acoustics, Speech and
  Signal Processing (ICASSP)}, pages 4016--4019. IEEE, 2011.

\bibitem{reise2010reconstruction}
G{\"u}nter Reise and Gerald Matz.
\newblock Reconstruction of time-varying fields in wireless sensor networks
  using shift-invariant spaces: Iterative algorithms and impact of sensor
  localization errors.
\newblock In {\em 2010 IEEE 11th International Workshop on Signal Processing
  Advances in Wireless Communications (SPAWC)}, pages 1--5. IEEE, 2010.

\bibitem{reise2012distributed}
G{\"u}nter Reise, Gerald Matz, and Karlheinz Grochenig.
\newblock Distributed field reconstruction in wireless sensor networks based on
  hybrid shift-invariant spaces.
\newblock {\em IEEE Transactions on Signal Processing}, 60(10):5426--5439,
  2012.

\bibitem{sun2007nonuniform}
Qiyu Sun.
\newblock Nonuniform average sampling and reconstruction of signals with finite
  rate of innovation.
\newblock {\em SIAM journal on mathematical analysis}, 38(5):1389--1422, 2007.

\bibitem{tang2017system}
Sui Tang.
\newblock System identification in dynamical sampling.
\newblock {\em Advances in Computational Mathematics}, 43:555--580, 2017.

\bibitem{tang2017universal}
Sui Tang.
\newblock Universal spatiotemporal sampling sets for discrete spatially
  invariant evolution processes.
\newblock {\em IEEE Transactions on Information Theory}, 63(9):5518--5528,
  2017.

\bibitem{ulanovskii2021reconstruction}
Alexander Ulanovskii and Ilya Zlotnikov.
\newblock Reconstruction of bandlimited functions from space--time samples.
\newblock {\em Journal of Functional Analysis}, 280(9):108962, 2021.

\bibitem{unnikrishnan2012sampling}
Jayakrishnan Unnikrishnan and Martin Vetterli.
\newblock Sampling high-dimensional bandlimited fields on low-dimensional
  manifolds.
\newblock {\em IEEE transactions on information theory}, 59(4):2103--2127,
  2012.

\bibitem{yao2023space}
Qing Yao, Longxiu Huang, and Sui Tang.
\newblock Space-time variable density samplings for sparse bandlimited graph
  signals driven by diffusion operators.
\newblock In {\em ICASSP 2023-2023 IEEE International Conference on Acoustics,
  Speech and Signal Processing (ICASSP)}, pages 1--5. IEEE, 2023.

\bibitem{zlotnikov2022planar}
Ilya Zlotnikov.
\newblock On planar sampling with gaussian kernel in spaces of bandlimited
  functions.
\newblock {\em Journal of Fourier Analysis and Applications}, 28(3):55, 2022.

\end{thebibliography}
\bibliographystyle{plain}

\end{document}